\documentclass[11pt, letterpaper]{article}
\title{Approaching the Soundness Barrier:\\ A Near Optimal Analysis of the Cube versus Cube Test}
\author{
Dor Minzer\thanks{Department of Mathematics, Massachusetts Institute of Technology, Cambridge, USA. Supported by a Sloan Research
Fellowship.}
\and
Kai Zheng\thanks{Department of Mathematics, Massachusetts Institute of Technology, Cambridge, USA. Supported by the NSF Graduate Research Fellowship DGE-1745302.}
}
\date{\vspace{-5ex}}
\usepackage{fullpage}
\usepackage{amsthm}
\usepackage{amsmath,amssymb,amsfonts,nicefrac}
\usepackage{xspace}
\usepackage{color}
\usepackage{url}
\usepackage{hyperref}
\usepackage{bm}
\usepackage{bbm}
\usepackage{times}

\usepackage{enumitem}

\newcommand{\p}{\Pr}
\newcommand{\approxparam}[1]{\,{\stackrel{{#1}}{\approx}}\,}

\DeclareMathOperator{\AG}{AffGras}
\DeclareMathOperator{\G}{Gras}
\DeclareMathOperator{\poly}{poly}
\DeclareMathOperator{\spa}{span}
\DeclareMathOperator{\var}{var}

\newcommand{\E}{\mathop{\mathbb{E}}}
\newcommand{\Ff}{\mathbb{F}}
\newcommand{\D}{\mathcal{D}}

\newcommand{\C}{\mathcal{C}}

\newcommand{\Cxys}{{\mathcal{C}}_{\substack{x,\sigma\\ y, \tau}}}
\newcommand{\Fxys}{{\mathcal{F}}_{\substack{x,\sigma\\ y, \tau}}}

\newcommand{\Fxy}{{\mathcal{F}}_{x,\sigma, y}}
\newcommand{\ind}{\mathbbm{1}}

\newcommand{\eps}{\varepsilon}
\renewcommand{\epsilon}{\eps}
\newcommand{\num}{n}
\newcommand{\Fm}{\mathbb{F}_q^{\num}}
\newcommand\skipi{{\vskip 10pt}}
\renewcommand\leq{\leqslant}
\renewcommand\geq{\geqslant}
\renewcommand\le{\leqslant}

\theoremstyle{plain} %% This is the default, anyway

\newtheorem{theorem}{Theorem}[section]
   \newtheorem{thm}{Theorem}[section]
   
   \newtheorem{lemma}[thm]{Lemma}

   \newtheorem{definition}{Definition}

\begin{document}
\maketitle
\begin{abstract}
    The Cube versus Cube test is a variant of the well-known Plane versus Plane test of Raz and Safra~\cite{RS}, in which to each $3$-dimensional affine subspace $C$ of $\mathbb{F}_q^n$,
    a polynomial of degree at most $d$, $T(C)$, is assigned in a somewhat locally consistent manner: taking two cubes $C_1, C_2$ that intersect in a plane uniformly
    at random, the probability that $T(C_1)$ and $T(C_2)$ agree on $C_1\cap C_2$ is at least some $\eps$. An element of interest is the soundness threshold of this test, i.e.\
    the smallest value of $\eps$,
    such that this amount of local consistency implies a global structure; namely, that there is a global degree $d$ function $g$ such that $g|_{C} \equiv T(C)$ for at least 
    $\Omega(\eps)$ fraction of the cubes.
    
    We show that the cube versus cube low degree test has soundness $\poly(d)/q$. This result achieves the optimal dependence on $q$ for soundness in low degree testing and improves upon previous soundness results of $\poly(d)/q^{1/2}$ due to Bhangale, Dinur and Navon~\cite{BDN}.
\end{abstract}
\section{Introduction}
\subsection{Low degree testing}
The Reed-Muller code is a basic building block of many results in Theoretical Computer Science. One of the features that makes it so useful, particularly in the area of Probabilistically
Checkable Proofs~\cite{FGLSS,AroraSafra,ALMSS}, is that it admits very efficient local tests. Most relevant to the current paper are the line versus line~\cite{RSu},
plane versus plane~\cite{RS}, and their $3$-dimensional analog, the cube versus cube test~\cite{BDN}. In these settings, we have a finite field $\mathbb{F}_q$ where $q$ is
thought of as large, a degree parameter $d\in\mathbb{N}$ much smaller than $q$, and we wish to encode a degree $d$ polynomial\footnote{Here and throughout, the notion of degree we refer to is the total degree of a polynomial.} $f\colon \mathbb{F}_q^n\to\mathbb{F}_q$ using an encoding scheme that allows for local testing with $2$ queries.

The most basic example of these encoding schemes and tests is given by the line versus line test. To define this, let $\mathcal{L}$ be the set of all lines in $\mathbb{F}_q^n$. A
polynomial $f\colon \mathbb{F}_q^n\to\mathbb{F}_q$ is thus encoded by the restrictions-to-lines table $T$ that assigns to each line $\ell\in \mathcal{L}$ the restriction of $f$ to
$\ell$, i.e.\ $f_{|{\ell}}$. The test that accompanies this encoding scheme is the line versus line test, described as follows:
\begin{enumerate}
    \item Choose a point $x \in \Ff_q^n$ uniformly at random.
    \item Choose two lines $\ell_1, \ell_2$ uniformly and independently conditioned on $\ell_1, \ell_2 \ni x$.
    \item Read $T(\ell_1)$ and $T(\ell_2)$ and check that $T(\ell_1)|_{x} = T(\ell_2)|_{x}$, i.e.\ that these two functions agree on $x$.
\end{enumerate}
It is clear that if $T$ is indeed a table of restrictions of a given polynomial $f$ of degree at most $d$, then the above test passes with probability $1$. The interesting question in
this context is the converse: suppose we have a table of functions $T$ that assigns to each line a function of degree at most $d$, and suppose that the above test passes with probability
at least $s>0$; is it necessarily the case that this table of functions is associated with some global degree $d$ polynomial?

Early works~\cite{FGLSS, AroraSafra, ALMSS} were only able to analyze this test in the case that the soundness parameter $s$ is close to $1$,
namely the case in which we are guaranteed that the table of functions $T$ passes the test with probability close to $1$, say $s = 1-\eps$. In this case, it is typically shown that
any such table is close (in Hamming distance) to a table $T'$ which is an actual table of restrictions of some polynomial $f'\colon\mathbb{F}_q^n\to\mathbb{F}_q$ of degree at most $d$.

To strengthen the PCP theorem however (and more precisely, to improve on the soundness guarantee), it became clear that one has to be able to analyze these tests for as small of
a soundness parameter $s$ as possible. Towards this end, better analysis of the line versus line test was given~\cite{AS} and other variants of the line versus line test were considered.
Most relevant to us are the higher dimensional analogs of the tests, which are the plane versus plane test~\cite{RS} and the cube versus cube test~\cite{BDN}. In these tests, instead
of encoding a polynomial $f\colon\mathbb{F}_q^n\to\mathbb{F}_q$ using its table of restrictions to lines, one encodes $f$ using its table of restrictions to planes (in the case of the
plane versus plane test) and to $3$-dimensional affine subspaces (in the case of the cube versus cube test).

In the cube versus cube test, a polynomial $f\colon\mathbb{F}_q^n\to\mathbb{F}_q$ is encoded using its cubes table $T$, which assigns to each affine cube $C \subseteq \Fm$ a degree at most $d$ polynomial $T(C)$ which is equal to the restriction $f_{|C}$. The associated test with this encoding scheme is the cube versus cube test, defined as:
\begin{enumerate}
    \item Choose an affine plane $P \subseteq \Fm$ uniformly at random.
    \item Choose two affine cubes $C_1, C_2$ uniformly at random such that $C_1, C_2 \supseteq P$.
    \item Read $T(C_1)$ and $T(C_2)$ from the cubes table and accept if and only if the restrictions to $P$ satisfy $T(C_1)|_{P} = T(C_2)|_{P}$.
\end{enumerate}
It is clear that a valid table of restrictions $T$ passes the test with probability $1$, and again the interesting question is the converse. Namely, suppose that a table of functions
$T$ passes the above test with probability $s>0$; is it necessarily the case that $T$ may be associated with some degree $d$ polynomial?

More generally, one can consider tests where affine cubes and planes are replaced with affine subspaces of dimensions $k$ and $\ell$ with $k > \ell$. For such tests, the table $T$ contains supposed restrictions of $f$ to each dimension $k$ affine subspace and performs the following dimension $\ell$ agreement test:
\begin{enumerate}
    \item Choose an affine dimension $\ell$ subspace $U \subseteq \Fm$ uniformly at random.
    \item Choose two affine dimension $k$ subspaces $V_1, V_2$ uniformly at random such that $V_1, V_2 \supseteq U$.
    \item Read $T(V_1)$ and $T(V_2)$ from the table and accept if and only if the restrictions to $U$ satisfy $T(V_1)|_{U} = T(V_2)|_{U}$.
\end{enumerate}

Indeed the cube versus cube test is the $k = 3, \ell = 2$ case, and the originally studied Raz-Safra plane versus plane test is the $k = 2, \ell = 1$ case \cite{RS}. As discussed earlier,
the primary reason these tests were considered is that they admit a very low soundness error. The soundness error of a test is defined to be the smallest $s$ such that if the test passes with probability at least $s$, then there exists a degree $d$ function $g$ such that $g_{V} = T(V)$ for an $\Omega(s)$-fraction of the dimension $k$ affine subspaces $V$.

For the plane vs plane test, Raz and Safra~\cite{RS} showed that $s\geq \frac{n^{C}d^C}{q^c}$ for some absolute constants $c,C>0$, a result that was sufficient in order to prove an improved
PCP characterization of NP (see for example~\cite{DFKRS}). An improved analysis of this test was given by Moshkovitz and Raz~\cite{MR}, who showed that $s\geq \poly(d)/q^{1/8}$. We note
that the natural lower bound on $s$ is $\Theta(1/q)$, since there are tables of assignments $T$ on which the test passes with probability $\Theta(1/q)$ yet all degree $d$ functions
agree with at most $o(1/q)$ of the entries of $T$. Indeed, take $h = c q^3$ for some small absolute constant $c>0$, and pick affine subspaces $W_1,\ldots,W_{h} \subseteq\mathbb{F}_q^n$
of dimension $n-1$ uniformly at random; for each $i=1,\ldots,h$, also pick a degree $d$ polynomial $f_i\colon W_i \to\mathbb{F}_q$. To define the table $T$, for each plane $P$ pick
the first $i$ such that $P\subseteq W_i$ if such $i$ exists and define $T(P) = f_i|_{P}$; otherwise, pick $T(P)$ randomly. It is not hard to see that if we pick planes, $P_1,P_2$, randomly that intersect in a line, and $P_1\subseteq W_i$, then $P_2\subseteq W_i$ with probability $\Omega(1/q)$. In this case, one can show that with constant probability
both $P_1,P_2$ are assigned by $f_i$, and hence in expectation, the test passes with probability $\Omega(1/q)$. It can be shown though, that as each $f_i$ is chosen randomly
and each $W_i$ contains $O(1/q^3)$ of the planes, that no degree $d$ polynomial agrees with $T$ on more than $o(1/q)$ of the planes.

In light of this, an intriguing open question is what the soundness threshold for the plane versus plane and other low degree tests is.
For the Raz-Safra plane versus plane test, the best known soundness analysis is still due to Moshkovitz and Raz in \cite{MR}. Motivated by this question and more recently,
Bhangale, Dinur, and Livni Navon~\cite{BDN} suggested to study the cube versus cube test, and managed to show that its soundness threshold is higher than that
known in the Raz-Safra test. Specifically, they showed that the soundness threshold of the cube versus cube test is $s\geq \poly(d)/q^{1/2}$, and
ask whether it is the case that the true soundness threshold of the cube versus cube test is linear in $1/q$. Our main result confirms that this is indeed the case.

More precisely, the cube versus cube test studied in \cite{BDN} is slightly different than the one presented above. The test they consider is the generalized test above
with $k = 3$ and $\ell = 0$, meaning the tester picks a point $x\in\mathbb{F}_q^{\num}$, and then independently two cubes $C_1,C_2$ that contain $x$, and tests that $T(C_1)$
and $T(C_2)$ assign to $x$ the same value. As shown in \cite{BDN} though, the acceptance probability of the various cube versus cube tests are all virtually the same for any cubes
table $T$, so one can essentially ignore this difference:
\begin{theorem} \label{th: related tests}
For a table $T$ giving restrictions to dimension $k \leq \frac{n}{2}$ subspaces, let $\alpha_{k\ell k}(T)$ denote the probability that $T$ passes the dimension $\ell$ agreement test. Then for $0 \leq r < r' < k$,
\begin{equation*}
     (1 - o(1))\alpha_{krk}(T) \leq \alpha_{kr'k}(T) \leq \alpha_{krk}(T) + (1 + o(1))q^{-(k-2r'+r+1)}.
\end{equation*}
\end{theorem}

In this paper, we improve the analysis in \cite{BDN} and show that the cube versus cube test has soundness $\frac{10^7d^6}{q}$. As mentioned, we think of $q$ as large, and $d$ as much smaller than $q$, but still larger than some absolute constant. It suffices to assume that $d < \frac{q^{1/9}}{11}$. Our main theorem is as follows:
\begin{theorem} \label{th: main}
Suppose $T$ is a cubes table such that for each affine cube $C \subseteq \Fm$, $T(C)$ is a degree $d$ polynomial over $C$. If $T$ passes the cube versus cube test with probability $\epsilon \geq \frac{10^7d^6}{q}$, then there exists a degree $d$ polynomial, $g$, such that $g|_{C} = T(C)$ for an $\Omega(\epsilon)$-fraction of all affine cubes $C$.
\end{theorem}

We note that while Theorem~\ref{th: related tests} implies our soundness bound also holds for the the cube versus cube test with point intersection considered in \cite{BDN}, the same is not true the other way around. Indeed, from Theorem~\ref{th: related tests} one cannot obtain any lower bound on $\alpha_{303}$ from $\alpha_{323}$. Thus, prior to this work even an $O(1/q^2)$ soundness result was not known for the cube versus cube test with dimension $2$ intersection.

We also remark that we made no attempt to optimize the dependence on $d$, and that as noted earlier the main point of this result is that the dependency on the field size $q$ is optimal.

\subsection{Proof Overview}
Our proof strategy relies heavily on ideas from~\cite{BDN}, but to get the optimal soundness threshold our argument uses more refined expansion
as well as sampling arguments.

\paragraph{High level description.}
Suppose that the cube versus cube test passes with probability $\epsilon \geq \frac{10^7d^6}{q}$. For a point $x\in \Fm$ and a value
$\sigma\in \mathbb{F}_q$, define
\[
\C_{x} = \{ C\subseteq \Fm~|~C\text{ is a $3$-dimensional cube, $x\in C$}\},
\qquad
\C_{x,\sigma} = \{ C\subseteq \Fm~|~C\in \C_x\text{ and $T(C)_{|x} = \sigma$}\}.
\]
In words, $\C_{x,\sigma}$ is the set of cubes containing $x$ on which $T$ gives the point $x$ the value $\sigma$. We may view $\C_{x,\sigma}$ as a partition
of $\C_{x}$ according to the value given to $x$. In~\cite{BDN}, the authors show that this partition cannot consist of only parts that are small, otherwise
the contribution to the acceptance probability of the test of these cube would be small. Our argument observes further properties of these partitions. At a high
level, we show that if we take $x$ and $y$ randomly and consider the partition $\C_{x,\sigma}\cap \C_{y,\tau}$ of the cubes containing
both $x$ and $y$, then this partition cannot consist only of small parts. This suggests that the partitions $\C_{x,\sigma}$ and $\C_{y,\tau}$ are correlated.
In fact, in a sense we are able to show that (in the part that contributes to the acceptance probability of the test), they are in $1$ to $1$ correspondence
in the sense that for each $\sigma$ we can identify a single $\tau$ so that $\C_{x,\sigma}\cap \C_y$ and $\C_x\cap \C_{y,\tau}$ are roughly the
same set. We then show that for a sufficiently large fraction of $x$'s and $\sigma$, there is a global function agreeing with $T$ on almost all of $\C_{x,\sigma}$. The above correlation between partitions then allows us to argue that the global functions of $x$ and $y$ agree provided the parts $\C_{x,\sigma}$ and
$\C_{y,\tau}$ are matching according to the above described $1$-to-$1$ correspondence. Using this, we are able to show that for some $x$, the global
function on $\C_{x,\sigma}$ agrees with sufficiently many global functions on $\C_{y,\tau}$ for other $y$'s, and enough so that this already gives agreement
$\Omega(\eps)$ with the table $T$. A more detailed overview follows.

\paragraph{A more detailed description.}
Consider the following way of generating a pair of cubes
intersecting in a plane: sample $x,y\in\Fm$, sample a plane $P$ containing both $x$ and $y$, and sample $C_1, C_2$ cubes containing $P$. We note that the points $x$ and
$y$ partition the set of cubes $C$ containing them into $q^2$ sets:
\[
\C_{\substack{x,\sigma \\ y,\tau}} := \C_{x,\sigma}\cap \C_{y,\tau},
\qquad \text{for all }\sigma,\tau\in\mathbb{F}_q.
\]
We further note that for the test on $C_1, C_2$ to pass, they both must belong to the same part in the partition, i.e. they must both be in $C_{\substack{x,\sigma \\ y,\tau}}$
for some $\sigma,\tau$. Thus, fixing $x$ and $y$, this shows that the probability of the test passing (conditioned on that) is related to the expansion properties of
the partition $\left(\C_{\substack{x,\sigma \\ y,\tau}}\right)_{\sigma,\tau\in\mathbb{F}_q}$ in the affine Grassmann graph. We formally define this graph in the
next section, but informally in our context this graph contains all of the cubes $C$ containing both $x$ and $y$, and two cubes are adjacent if they intersect in a plane.
Thus, we get that the probability that the test passes is at most
\[
\sum\limits_{\sigma,\tau}
\Pr_{C_1\cap C_2 \supseteq P \ni x,y}[C_1, C_2\in \C_{\substack{x,\sigma \\ y,\tau}}]
=
\sum\limits_{\sigma,\tau}
\mu_{x,y}\left(\C_{\substack{x,\sigma \\ y,\tau}}\right)
\left(1-\Phi\left(\C_{\substack{x,\sigma \\ y,\tau}}\right)\right).
\]
Here, $\mu_{x,y}\left(\C_{\substack{x,\sigma \\ y,\tau}}\right)$ is the relative measure of $\C_{\substack{x,\sigma \\ y,\tau}}$ among all cubes containing $x,y$
and $\Phi_{xy}\left(\C_{\substack{x,\sigma \\ y,\tau}}\right)$ is the edge expansion of $\C_{\substack{x,\sigma \\ y,\tau}}$ in the above defined Grassmann graph.
Using the spectral properties of the Grassmann graph, we show that this sum is at most $O\left(\frac{1}{q}\right) + \sum\limits_{\sigma,\tau}
\mu_{x,y}\left(\C_{\substack{x,\sigma \\ y,\tau}}\right)^2$, and thereby conclude (using the fact that the probability the cube versus cube test passes with probability
at least $\eps$) that
\begin{equation}\label{eq:intro}
    \E_{x,y}\left[ \sum_{\sigma, \tau \in \Ff_q} \mu_{x,y}\left( \Cxys \right)^2 \right] \geq \epsilon.
\end{equation}
Intuitively, this says that the partitions given by $(C_{x,\sigma})_{\sigma\in\mathbb{F}_q}$ and $(C_{y,\tau})_{\tau\in\mathbb{F}_q}$ are somewhat correlated
with each other. To make use of this correlation, we must pass first to subsets of $C_{x,\sigma}$ on which we already know that we have some global structure.

For that, we use ideas from \cite{BDN}. Specifically, we identify and consider the pairs $(x,\sigma)$ that contribute almost all of the acceptance probability of the test, and show that for
each such $(x,\sigma)$ one may find a global degree $d$ function $g_{x,\sigma}$ that agrees with $T$ on almost all of $\C_{x,\sigma}$. Thus, letting $\mathcal{F}_{x,\sigma}$ be the set
of cubes containing $x$ on which $T$ and $g_{x,\sigma}$ agree, we get that $\mathcal{F}_{x, \sigma} \subseteq \C_{x, \sigma}$ are very close to each other. This motivates us to
define $\mathcal{F}_{\substack{x, \sigma\\ y, \tau}} = \mathcal{F}_{x, \sigma} \cap \mathcal{F}_{y, \tau}$. Since $\mathcal{F}_{x,\sigma}$ and $\C_{x,\sigma}$ are very close to each other, one expects~\eqref{eq:intro} to imply something
similar about the partial partition $\mathcal{F}_{\substack{x, \sigma\\ y, \tau}}$, and we show that this in indeed the case (though not quite as obviously as one may initially expect). Namely, we
show that~\eqref{eq:intro} implies that
\begin{equation}\label{eq:intro2}
    \E_{x,y}\left[ \sum_{\sigma, \tau \in \Ff_q} \mu_{x,y}\left( \Fxys \right)^2 \right] \geq \Omega(\epsilon).
\end{equation}
We note now that while the sum is over $q^2$ terms, for each $\sigma$ there exists at most a single $\tau$ for which the corresponding summand is non-zero. Indeed,
fixing $\sigma$ means that we look at $\mathcal{F}_{x,\sigma}$ on which we have some global degree $d$ function $g_{x,\sigma}$, and hence the only viable option for $\tau$
is $\tau = g_{x,\sigma}(y)$. Moreover, by averaging considerations it turns out that we can have at most $O(1/\eps)$ many $\sigma$'s for which the pair $(x,\sigma)$
contributes to the above sum. Hence we may find $x$ and $\sigma$ such that
\[
\E_{y}\left[\mu_{x,y}\left( \mathcal{F}_{\left\{\substack{x,\sigma \\ y, g_{x,\sigma}(y)}\right\}} \right)^2\right] \geq \Omega(\eps^2).
\]
Let $\eta = \mu_x(\mathcal{F}_{x,\sigma})$. Sampling $y\in\mathbb{F}_q^{\num}$ uniformly, it is clear that the expected value of
$\mu_{x,y}(\mathcal{F}_{x,\sigma}\cap \mathcal{C}_y)$ is $\eta$, and in fact one can show that a relatively strong concentration
holds. This concentration is strong enough to show that the fraction of $y$'s for which $\mu_{x,y}(\mathcal{F}_{x,\sigma}\cap \mathcal{C}_y)\geq 2\eta$
is very small -- small enough so that discarding them from the above inequality only incurs a small loss. Furthermore, we can also neglect $y$'s such that
$\mu_{x,y}\left( \mathcal{F}_{\left\{\substack{x,\sigma \\ y, g_{x,\sigma}(y)}\right\}} \right)\leq c\eps$ for some sufficiently small absolute constant $c$, so altogether we get that
\[
\E_{y}\left[\mu_{x,y}\left( \mathcal{F}_{\left\{\substack{x,\sigma \\ y, g_{x,\sigma}(y)}\right\}} \right)^2 \ind_{\left\{c\eps\leq \mu_{x,y}\left( \mathcal{F}_{\left\{\substack{x,\sigma \\ y, g_{x,\sigma}(y)}\right\}} \right)\leq 2\eta\right\}}\right] \geq \Omega(\eps^2).
\]
Doing a dyadic partitioning, we get that there is a $t \geq c\eps$  such that
\[
t^2 p_t\geq \Omega\left(\frac{\eps^2}{\log(1/\eps)}\right),
\qquad\text{ where }~~
p_t = \Pr_{y}\left[\mu_{x,y}\left( \mathcal{F}_{\left\{\substack{x,\sigma \\ y, g_{x,\sigma}(y)}\right\}} \right) \in [t,2t)\right].
\]
Taking $Y$ to be the set of $y$'s such that $\mu_{x,y}\left( \mathcal{F}_{\left\{\substack{x,\sigma \\ y, g_{x,\sigma}(y)}\right\}} \right)\in [t,2t)$,
we show $g_{x,\sigma} \equiv g_{y,\tau}$ for all $y \in Y$. Hence, $g_{x,\sigma}$ agrees with the table $T$ on all of the cubes
in $\mathcal{F} = \cup_{y\in Y} \mathcal{F}_{y,\tau}$, and to complete the proof it suffices to show that $\mu(\mathcal{F})\geq \Omega(\eps)$. This is done by
applying a standard spectral argument and using the fact that $\mu(Y) = p_t \geq \Omega\left(\frac{\eps^2}{t^2 \log(1/\eps)}\right)$.
\section{Preliminaries}
\subsection{Notations}
Throughout the paper, we let $q$ denote the field size, $\num$ denote the dimension, and let
$\mathcal{C}$ be the set of all \emph{affine} cubes in $\mathbb{F}_q^{\num}$. An affine cube is a linear $3$-dimensional subspace with all points translated by some $x_0 \in \Ff_q^n$. We henceforth refer to affine cubes as simply cubes.

For points $x,y\in\mathbb{F}_q^{\num}$, we denote
\[
\mathcal{C}_x = \{C \in \mathcal{C} \; | \; x \in C \},
\qquad \qquad
\mathcal{C}_{x,y} = \mathcal{C}_x\cap \mathcal{C}_y.
\]
We also let $\mathcal{L}$ denote the set of all lines in $\Ff_q^n$ and likewise let $\mathcal{L}_x$ denote the set of all lines containing $x$. We let $\mu$ be the uniform distribution over $\mathcal{C}$, $\mu_x$ be the uniform measure over $\mathcal{C}_x$ and $\mu_{x,y}$ be the uniform measure over
$\mathcal{C}_{x,y}$. Abusing notation, we will sometimes use $\mu$ to also denote the uniform measure over $\mathbb{F}_q^{\num}$, and it will be clear from context if we are referring
to the uniform measure over $\mathcal{C}$ or over $\mathbb{F}_q^{\num}$. 

For two functions $f$ and $g$, we say that $f$ is $\delta$-close to $g$ if they  differ on at most a $\delta$ fraction of their inputs. That is, $\Pr_{x}[f(x) \neq g(x)] \leq \delta$, where the space of inputs $x$ is the domain of $f$ and $g$. We will also write $f \approxparam{\delta} g$ to denote that $f$ is $\delta$-close to $g$.

Finally, we will have a cubes table $T$ which assigns to each cube $C\in \mathcal{C}$ a degree $d$ polynomial $T(C)\colon C\to\mathbb{F}_q$. Given such table,
we define
\[
\mathcal{C}_{x,\sigma} = \{C \in \mathcal{C}_x \; | \; T(C)(x) = \sigma \},
\qquad\qquad
\mathcal{C}_{\substack{x,\sigma\\ y,\tau}} = \mathcal{C}_{x,\sigma} \cap \mathcal{C}_{y,\tau},
\]
and we denote for convenience $\eps(T) = \alpha_{323}(T)$, i.e.\ the probability that the cube versus cube test on $T$ passes. 

\skipi
Henceforth, we suppose that $T$ is a cubes table that passes the cube versus cube test with probability $\eps(T) = \epsilon \geq \frac{10^7d^6}{q}$.

% Now fix $\epsilon = d^6/q$, $\gamma = 1/d^3$, and let $\epsilon$ be the exact pass probability of the test.

\subsection{Grassman and Affine Grassman Graphs}

In this section we introduce the Grassman and Affine Grassman graphs. At times it will be helpful to think about the cube versus cube test in terms of these graphs and use known results about expansion in these graphs.

\begin{definition}
The \emph{Grassman graph} $\G(k, \ell)$ is the graph with vertex set consisting of all $k$-dimensional linear subspaces and edges between all pairs of subspaces $(U_1, U_2)$ satisfying $\dim (U_1 \cap U_2) = \ell$.
\end{definition}
\begin{definition}
The \emph{Affine Grassman graph} $\AG(k, \ell)$, is the graph with vertex set consisting of all $k$-dimensional affine subspaces and edges between all pairs of subspaces $(V_1, V_2)$ satisfying $\dim (V_1 \cap V_2) = \ell$.
\end{definition}

If we let $G = \AG(3, 2)$ we can think of the cube versus cube test as choosing a random edge $(C_1, C_2) \in E$ in the graph and checking if $T(C_1)$ and $T(C_2)$ agree on their intersection. Thus we may let $S$ be the set of edges $(C_1, C_2) \in E$ where $T(C_1)|_{C_1 \cap C_2}=T(C_2)|_{C_1 \cap C_2}$. Since $\epsilon$ is the exact pass probability of the test, we have
\begin{equation*}
    \Pr\left[(C_1, C_2) \in S\right] = \epsilon.
\end{equation*}

If we fix a point $x$ and condition on $C_1, C_2,$ and $C_1 \cap C_2$ containing $x$, then notice that with the remaining degrees of freedom, choosing $C_1$ and $C_2$ corresponds to choosing two linear cubes that intersect in a plane. Equivalently, this is choosing a random edge from the induced subgraph of $G$ on $\mathcal{C}_x$, which we denote by $G_x$. Notice that $G_x$ is isomorphic to $\G(3,2)$. If we further condition on $C_1, C_2,$ and $C_1 \cap C_2$ containing two points $x$ and $y$, then the random choice of $C_1$ and $C_2$ is equivalent to choosing two linear planes that intersect in a line, or a random edge in $G_{xy}$ which is isomorphic to $\G(2,1)$.

If we fix a point $x$ and condition on $C_1 \in \mathcal{C}_{x,\sigma}$, then notice that the test can only pass if $C_2 \in \mathcal{C}_{x,\sigma}$ as well. This is where expansion in the Grassman graphs plays a role. Conditioned on $C_1 \in \mathcal{C}_{x,\sigma}$, and $C_2$ intersecting $C_1$ in a plane containing $x$, the probability that
$C_2 \in \mathcal{C}_{x, \sigma}$ is precisely $1- \Phi_x(\mathcal{C}_{x,\sigma})$, where $\Phi_x$ denotes expansion in $G_x$. Likewise, if we fix two points $x, y$, and condition on $C_1 \in \Cxys$, then we can lower bound the pass probability by $1 - \Phi_{xy} \left( \Cxys \right)$, which is the probability that $C_2 \in \Cxys$. The following fact about expansion in $\G(k, k-1)$ allows us to freely convert these quantities to measures.

\begin{lemma} \label{lm: expansion}
Let $A$ be a set of $k$-dimensional subspaces over $\Fm$ and let $\Phi(A)$ denote the expansion of $A$ in $\G(k, k-1)$. Then
\begin{equation*}
    \mu(A)-\frac{1}{q^{n-k+1} - 2} \leq
    1 - \Phi(A) \leq \mu(A) + \frac{1}{q}.
\end{equation*}
\end{lemma}
\begin{proof}
Let $F$ denote the indicator function for $A$. Thus $F$ is a function that takes $k$-dimensional subspaces as input and is $1$ on those in $A$ and $0$ otherwise. Let $T$ denote the adjacency operator of $\G(k, k-1)$. By~\cite[Theorem 9.4.1]{GM}, the eigenvalues of $T$ are
\[
\lambda_r =
\frac{-(q^{k} - 1) + q^r(q^{k-r}-1)(q^{n-k-r+1} - 1)}{(q^{k} - 1)(q^{n-k+1}-2)},
\]
for $r=0,1\ldots,k$. It is easily seen that for $r=0$ we get an eigenvalue of $1$, and for any other $r$ the corresponding eigenvalue is at most $q^{-r}$.
Thus, we may write $F = F_0 + F_1+\ldots+F_k$ where $F_r$ is an eigenvector of $T$ of eigenvalue $\lambda_r$, where $F_0\equiv \mu(S)$, and so
\begin{equation*}
    1-\Phi(A) = \frac{\langle F, TF \rangle}{\mu(S)} = \frac{1}{\mu(A)}\left(\mu(A)^2 + \sum_{r\geq 1} \lambda_r \|F_r\|_2^2 \right)
    \leq \mu(A) + \frac{1}{q}\frac{\sum_{r\geq 1} \|F_r\|_2^2}{\mu(A)}
    \leq \mu(A) + \frac{1}{q},
\end{equation*}
where we use the fact that $\sum_{r\geq 1} \|F_r\|_2^2\leq \sum_{r\geq 0} \|F_r\|_2^2 = \|F\|_2^2 = \mu(A)$. 
For the lower bound, we note that the eigenvalues $\lambda_r$ are nonnegative for $r < k$, while $\lambda_{k} = \frac{-1}{q^{n-k+1} - 2}$.
\end{proof}
\subsection{Bipartite Inclusion Graphs}
Next, we review some results regarding sampling edges in bipartite inclusion graphs that we will use later. These results were used in the context of direct product testing in \cite{IKW} and to analyze the cube versus cube test in \cite{BDN}.

Let $G(A,B)$ denote the bipartite inclusion graph with vertices $A \cup B$ and edges $E$, where $A$ consists of subspaces of some dimension $i$ (possibly $i = 0$ if $A$ is a set of points) in $\Ff_q^m$, $B$ consists of subspaces of some dimension $j > i$ in $\Ff_q^m$, and $(a,b) \in E$ if $a \subseteq b$. Also let $\lambda(G)$ denote the second largest singular value of a graph $G$ and for any vertex $v$ in the graph, let $N(v)$ denote its neighborhood. It is known that such graphs satisfy the following sampling property:
\begin{lemma} \label{lm: edge sampling}
Let $G = G(A,B)$ be a bi-regular bipartite graph. For every subset $B'\subset B$ of measure $\mu>0$ and every $E'\subset E$
\[ \left| \Pr_{\substack{b \in B'\\ a \in N(b)}}[(a,b)\in E'] - \Pr_{\substack{a \in A\\ b \in N(a) \cap B'}}[(a,b)\in E'] \right| \leq \frac{\lambda(G)}{\sqrt{\mu}} .\]
%Where it is understood that if $D_2$ output , we treat it as if $(a,b)\not\in E'$.
\end{lemma}
For completeness, we list the singular values of the bipartite inclusion graphs that we will use. 

\begin{lemma} \label{lm: singular value}
For $m\geq 6$, the bipartite inclusion graphs for subspaces of $\Ff_q^m$ have the following second singular values:
\begin{enumerate}
    \item For $G_1 = G(\mathcal{L} \setminus\mathcal{L}_x, \C_x)$, $\lambda(G_1) \approx \frac{1}{\sqrt{q}}$.
    \item For $G_2 = G(\mathcal{L}_x, \C_x)$, $\lambda(G_2) \approx \frac{1}{q}$.
    \item For $G_3 = G(\Ff_q^m \setminus \{x\}, \C_{x})$, $\lambda(G_3) \approx \frac{1}{q}$.
    \item For $G_4 = G(\Ff_q^m \setminus \ell, \C_{\ell})$, $\lambda(G_4) \approx \frac{1}{\sqrt{q}}$.
    \item For $G_5 = G(\Ff_q^m, \C)$, $\lambda(G_5) \approx \frac{1}{q^{3/2}}$.
    \item For $G_6 = G(\Ff_q^3, \mathcal{L})$, $\lambda(G_6) \approx \frac{1}{\sqrt{q}}$.
\end{enumerate}
We use $\approx$ to denote equality up to a multiplicative factor of $1 + o(1)$, where $o(1)$ is a function of $q$ that approaches $0$ as $q \xrightarrow[]{} \infty$.
\end{lemma}

We refer the reader to the appendix of \cite{BDN} for the proofs of Lemmas~\ref{lm: edge sampling} and \ref{lm: singular value}.

\section{The Main Argument}
\subsection{Local Agreement with Degree d polynomials} \label{sec: local}
In this section we show how to construct local degree $d$ polynomials $g_x$, that agree with an $\Omega(\epsilon)$ fraction of the cubes in $C_x$. The strategy is similar to that of \cite{BDN}. The following notion of excellent $(x, \sigma)$ pairs will be relevant to our analysis. This notion was introduced in an analysis of direct product tests \cite{IKW} and also used in \cite{BDN}.

Recall that $\epsilon \geq \frac{10^7d^6}{q}$ is the exact pass probability of the test. Let $\gamma$  be a small constant factor times $d^{-3}$, say $\gamma = \frac{1}{1000d^3}$.
\begin{definition} \label{def: excellent}
Call a pair $(x, \sigma)$ excellent if the following are true:
\begin{itemize}
    \item $\mu_x(\mathcal{C}_{x,\sigma}) \geq \epsilon / 5$
    \item For $C_1 \in C_{x, \sigma}$ chosen uniformly at random, a random line $\ell \subset C_1$ containing $x$, and $C_2 \in \mathcal{C}_{x, \sigma}$ chosen uniformly conditioned on containing $\ell$, we have that
    \[
    \Pr_{C_1, \ell, C_2}[T(C_1)|_{\ell} \neq T(C_2)|_{\ell}] \leq \gamma.
    \]
\end{itemize}
\end{definition}
Henceforth we let $X$ denote the set of all excellent pairs $(x, \sigma)$. Notice that if we let
\begin{equation*}
    p_{x, \sigma} = \Pr_{\substack{C_1 \in \mathcal{C}_{x,\sigma}, \\ C_2 \cap C_1 = \ell \ni x}}[T(C_1)(x) = T(C_2)(x), \; T(C_1)|_{\ell} \neq T(C_1)|_{\ell}],
\end{equation*}
then the second condition holds if $p_{x, \sigma} \leq \gamma \left(1-\Phi_{x}(\C_{x,\sigma})\right)$, since 
    \[
    \Pr_{C_1, \ell, C_2}[T(C_1)|_{\ell} \neq T(C_2)|_{\ell}] = \frac{p_{x,\sigma}}{1-\Phi_{x}(\C_{x,\sigma})} 
    \]
where the probability is with respect to the distribution in the second point of Definition~\ref{def: excellent}, and the inequality is due to Lemma~\ref{lm: expansion}.

This leads to the following observation:
\begin{lemma} \label{lm: pxsig}
If $(x, \sigma) \notin X$, then either $\mu_x(\mathcal{C}_{x,\sigma}) < \epsilon / 5$ or $p_{x,\sigma}  >  \gamma \left(1-\Phi_{x}(\C_{x,\sigma})\right)$.
\end{lemma}

The main result of this section shows that for an excellent $(x, \sigma)$, there exists a degree $d$ polynomial $g$ that agrees with almost all of the cubes in $\mathcal{C}_{x, \sigma}$. The same result has been proved in~\cite{BDN} under a stronger assumption on the largeness of $\eps$.
We begin by defining the function $f_{x, \sigma}$ by plurality over all cubes in $\mathcal{C}_{x, \sigma}$.
\begin{definition}
For a pair $(x, \sigma)$, define $f_{x,\sigma}: \Fm \xrightarrow[]{} \Ff_q$ so that $f_{x, \sigma}(y)$ is the most common value of $T(C)(y)$ over cubes $C \in \mathcal{C}_{x, \sigma}$ containing $y$. If there are no such cubes, assign the value arbitrarily.
\end{definition}

Since $(x, \sigma)$ is excellent, it follows immediately that $f_{x, \sigma}$ usually agrees with $T$ on points in cubes $C \in C_{x, \sigma}$.

\begin{lemma} \label{lm: local agree}
For $f_{x, \sigma}$ as defined above,
\[
\p_{C \in \mathcal{C}_{x, \sigma}, y \in C} [f_{x,\sigma}(y) = T(C)(y)] \geq 1- \gamma.
\]
% As a result, for at least half of the cubes $C \in \mathcal{C}_{x, \sigma}$, $f_{x,\sigma|C}$ and $T(C)$ agree on at least a $1-2\gamma$ fraction of $C$.
\end{lemma}

\begin{proof}
Let
\[
\gamma_y = \p_{C_1, C_2 \in \mathcal{C}_y \cap \mathcal{C}_{x, \sigma}} [T(C_1)(y) \neq T(C_2)(y)].
\]
Since $f_{x,\sigma}(y)$ is the most common value of $T(C)(y)$ over cubes $C \in \mathcal{C}_y \cap \mathcal{C}_{x, \sigma}$. We can write,

\[
    1-\gamma_y = \sum_{\alpha \in \Ff_q} \p_{C \in \mathcal{C}_y \cap \mathcal{C}_{x, \sigma}}[T(C)(y) = \alpha]^2
    \leq  \p_{C \in \mathcal{C}_y \cap \mathcal{C}_{x, \sigma}}[T(C)(y) = f_{x, \sigma}(y)].
\]

Since this inequality holds for any $y$, it also holds under expectation over any distribution of $y$.
Take the distribution as: choose $C\in \mathcal{C}_{x,\sigma}$ uniformly, and then $y\in C$ uniformly different from $x$
(or equivalently, a uniformly random line $\ell\subseteq C$ containing
$x$ and then $y\in \ell$ uniformly different from $x$). Then, denoting by $\ell_{x,y}$ the line that
passes through $x$ and $y$, we get that
\begin{align*}
    \p_{C \in \mathcal{C}_{x, \sigma}, y\in C}[T(C)(y) = f_{x,\sigma}(y)] \geq \mathbb{E}_{C, y}[1-\gamma_y]
    &= \p_{\substack{C_1 \in \mathcal{C}_{x, \sigma}, y\in C_1 \\ C_2\in \mathcal{C}_{x, \sigma}\cap \mathcal{C}_y}}[T(C_1)(y) = T(C_2)(y)]\\
    &\geq
    \p_{\substack{C_1 \in \mathcal{C}_{x, \sigma}, y\in C_1 \\ C_2\in \mathcal{C}_{x, \sigma}\cap \mathcal{C}_y}}[T(C_1)|_{\ell_{x,y}} = T(C_2)|_{\ell_{x,y}}]\\
     &\geq 1-\gamma,
\end{align*}
where the last inequality is due to the assumption that $(x, \sigma)$ is excellent.
\end{proof}

Finally, we show that $f_{x, \sigma}$ is close to a degree $d$ polynomial by using the following robust characterization of low degree polynomials due to Rubinfeld and Sudan.

\begin{theorem}\label{thm:RSu}\emph{\cite{RSu}}
Let $f: \Ff_q^m \xrightarrow[]{} \Ff_q$, and let $N_{y,h} = \{y + i(h-y)\; | \; i \in [d+1] \}$. If,
\begin{equation*}
    \p_{y,h \in \Fm}[\; \exists p \text{ such that } p|_{N_{y,h}} = f|_{N_{y,h}}] \geq 1-\delta,
\end{equation*}
for $\delta \leq \frac{1}{2(d+2)^2}$, then $f$ is $2\delta$-close to some degree $d$ polynomial $g$.
\end{theorem}

%\textcolor{red}{Define $\approxparam{\gamma}$ in the notation section. Define what closeness of functions is.}
\begin{lemma} \label{lm: close to d}
If $\p_{C \in \mathcal{C}_x }[T(C) \approxparam{2\gamma} f|_{C}] \geq \frac{\epsilon}{10}$, then $f$ is $4d\gamma$ close to some degree $d$ polynomial $g$.
\end{lemma}
\begin{proof}
Let $\mathcal{F} = \{ C \in \C_x \; | \; T(C) \approxparam{2\gamma} f|_{C} \}$. Fix $C \in \mathcal{F}$. We first show that for almost all lines in $C$, $f$ agrees with the degree $d$ function $T(C)$ on almost all points on the line. We do this using Lemma~\ref{lm: edge sampling} and the spectral properties of 
$G_C = G(A \cup B, E)$ -- the bipartite inclusion graph where $A$ is the set of all points in $C$, $B$ is the set of all affine lines in $C$, and $E$ is the set of all point-line pairs $(x, \ell)$ such that $x \in \ell$. This graph is the same as $G_{5}$ in Lemma~\ref{lm: singular value}, with $\Ff_q^3 \cong \C$, so $\lambda(G_C) \leq \frac{2}{\sqrt{q}}$.

Let $A' = \{ y \in A \; | \; T(C)(y) \neq f(y) \}$, and let $B' = \{ \ell \in B \; | \; |N(\ell) \cap A'| \geq 3\gamma|N(\ell)| \}$, where $N(\ell)$ consists of all of the points contained in $\ell$. In other words, $A$ is the set of points in $C$ where $f$ disagrees with $T(C)$, and $B'$ is the set of lines where $f$ and $T(C)$ disagree on at least $3\gamma$ of the points. Then, by Lemma~\ref{lm: edge sampling}
\begin{equation*}
   \left|\p_{\ell \in B', y \in N(\ell)} [y \in A'] - \p_{y \in A, y \in C \cap B'}[y \in A']\right| \leq \frac{\lambda(G_C)}{\sqrt{\frac{|B'|}{|B|}}}.
\end{equation*}
By construction, $\p_{\ell \in B', y \in N(\ell)} [y \in A'] \geq 3 \gamma$ and by assumption
\[
\p_{y \in A, C \in N(y) \cap B'}[y \in A'] \leq \p_{y \in A}[y \in A'] \leq 2\gamma,
\]
so $|B'| \leq \left(\frac{\lambda(G_C)}{\gamma}\right)^2 |B| \leq \gamma|B|$. Thus, for any cube $C \in \mathcal{F}$, $f$ is $3\gamma$-close to a degree $d$ polynomial on at least a $1-\gamma$ fraction of the lines in $C$. We next show that $\mathcal{F}$ contains enough cubes to cover nearly all lines. 

Consider the bipartite inclusion graphs $G_1 = G(\mathcal{L}\setminus\mathcal{L}_x, \C_x)$ and $G_2 = G(\mathcal{L}_x, \C_x)$, which have second singular values at most $\frac{1}{\sqrt{q}}$ and $\frac{1}{q}$ respectively by Lemma~\ref{lm: singular value}. Let $E'$ denote the set of line-cube pairs $(\ell, C)$ such that $\ell \subseteq C$ and $T(C)|_{\ell} \approxparam{3\gamma} f|_{\ell}$. Recall that $|\mathcal{F}|/|\C_x| \geq \epsilon/10$, so by Lemma~\ref{lm: edge sampling} on $G_1$,

\begin{equation*}
    \left|\p_{\ell, C \in N(\ell) \cap F}[(\ell, C) \in E'\; | \; \ell \notin \mathcal{L}_x] - \p_{C \in F, \ell \in C}[(\ell, C) \in E' \; | \; \ell \notin \mathcal{L}_x]\right|
    \leq \frac{1/\sqrt{q}}{\sqrt{\eps/10}}
    =
    \sqrt{\frac{10}{q\epsilon}} \leq \gamma.
\end{equation*}
By Lemma~\ref{lm: edge sampling} on $G_2$,
\begin{equation*}
    \left|\p_{\ell, C \in N(\ell) \cap F}[(\ell, C) \in E'\; | \; \ell \in \mathcal{L}_x] - \p_{C \in F, \ell \in C}[(\ell, C) \in E' \; | \; \ell \in \mathcal{L}_x]\right|
    \leq \frac{1/q}{\sqrt{\eps/10}}
    =
    \sqrt{\frac{10}{q^2\epsilon}} \leq \gamma.
\end{equation*}

We showed previously that for every $C \in \mathcal{F}$, $\p_{\ell \in C}[(\ell, C) \in E'] \geq 1 - \gamma$. Let $p = |\mathcal{L}_x| / |\mathcal{L}|$ be the probability that a randomly chosen line contains $x$. Then,
\begin{align*}
    \p_{\ell}[\exists C \text{ s.t. } (\ell, C) \in E'] &\geq \p_{\ell, C \in N(\ell) \cap F}[(\ell, C) \in E'] \\
    &= (1-p) \p_{\ell, C \in N(\ell) \cap F}[(\ell, C) \in E'\; | \; \ell \notin \mathcal{L}_x] + p  \p_{\ell, C \in N(\ell) \cap F}[(\ell, C) \in E'\; | \; \ell \in \mathcal{L}_x] \\
    &\geq (1-p) \p_{C \in F, \ell \in C}[(\ell, C) \in E' \; | \; \ell \notin \mathcal{L}_x] + p  \p_{C \in F, \ell \in C}[(\ell, C) \in E' \; | \; \ell \in \mathcal{L}_x] - \gamma \\
    &= \p_{\ell \in C}[(\ell, C) \in E'] -\gamma \\
    &\geq 1 - 2\gamma.
\end{align*}

Therefore, $f$ is $3\gamma$-close to a degree $d$ polynomial on at least $1-2\gamma$ of the lines in $\mathcal{L}$. This allows us to bound the probability that, for a randomly chosen neighborhood, $f$ is equal to a low degree polynomial on that neighborhood. To do this, we view picking a random neighborhood as randomly picking a line $\ell \in \mathcal{L}$, and then $y, h \in \ell$ uniformly at random.

\begin{align*}
    \p_{y, h \in \Fm}[\exists C \text{ s.t. } f|_{N_{y,h}} = T(C)|_{N_{y,h}}]
    &\geq
    \p_{\ell}[\exists C \text{ s.t. } (\ell, C) \in E'] \\
    &\cdot
    \p_{\ell, y,h \in \ell}[ f|_{N_{y,h}} = T(C)|_{N_{y,h}}\; | \;\exists C \text{ s.t. } (\ell, C) \in E']\\
    &\geq (1-2\gamma)(1-(d+2)3\gamma) \\
    &\geq 1-5d\gamma.
\end{align*}
Since $\gamma = \frac{1}{1000d^3}\leq \frac{1}{10(d+2)^3}$. Theorem~\ref{thm:RSu} now implies that there is a degree $d$ polynomial $g$ such that $f$ is $4d\gamma$ close to $g$.
\end{proof}

Combining Lemmas \ref{lm: local agree} and \ref{lm: close to d} we get that for each excellent $(x,\sigma)$, the plurality vote function $f_{x,\sigma}$ is $4d\gamma$-close to a degree $d$ polynomial.
\begin{lemma} \label{lm: gxsigma close}
For each excellent $(x, \sigma)$, the function $f_{x,\sigma}$ is $4d\gamma$-close to a degree $d$ polynomial.
\end{lemma}
\begin{proof}

Using Markov's inequality and Lemma~\ref{lm: local agree}, $\p_{C \in C_x}[T(C) \approxparam{2\gamma} f_{x,\sigma}|_{C}] \geq \frac{1}{2}\mu_x(\C_{x,\sigma})\geq \eps/10$. Thus, $f_{x,\sigma}$ satisfies the conditions of Lemma~\ref{lm: close to d}, and as a result, it is $4d\gamma$-close to some degree $d$ polynomial.
\end{proof}

\begin{lemma} \label{lm: gxsigma}
If $(x,\sigma)$ is excellent, then there exists a degree $d$ polynomial $g_{x, \sigma}$ such that $g_{x,\sigma}$ agrees with $T(C)$ for at least $1-2\sqrt{\gamma}$ of the cubes $C \in C_{x, \sigma}$.
\end{lemma}
\begin{proof}
Let $g_{x,\sigma}$ be the degree $d$ polynomial guaranteed by Lemma~\ref{lm: gxsigma close} and define,
\[
\mathcal{F} = \{C \in \C_x \; | \; f_{x,\sigma}|_{C} \approxparam{\sqrt{\gamma}} T(C) \},
\qquad
A' = \{ z \in \Fm \; | \; f_{x,\sigma}(z) \neq g_{x,\sigma}(z) \}.
\]
Then $\mu(A') \leq 4d\gamma$ and and by Markov's inequality and Lemma~\ref{lm: local agree}, $\mu_x(\mathcal{F}) \geq (1 - \sqrt{\gamma})\mu_x(\C_{x,\sigma}) \geq \frac{\epsilon}{10}$.

The bipartite inclusion graph $G_3 = G(\Fm \setminus\{x\}, \C_x)$ has second singular value at most $1/q$ by Lemma~\ref{lm: singular value}. By Lemma~\ref{lm: edge sampling} we have,

\begin{equation*}
    \left| \Pr_{C \in \mathcal{F}, y\in N(C)} [y \in A'] - \Pr_{y \in \Fm \setminus x, C \in N(y) \cap \mathcal{F}}[y \in A'] \right|
    \leq \frac{1/q}{\sqrt{\mu_x(\mathcal{F})}}
    \leq
    \frac{1/q}{\sqrt{\epsilon/10}} \leq 4\gamma.
\end{equation*}

Since $\Pr_{y \in \Fm \setminus\{x\}, C \in N(y) \cap \mathcal{F}}[y \in A'] \leq 2\mu(A')$, where the factor of $2$ accounts for the fact that $y \neq x$, it follows that $\Pr_{C \in \mathcal{F}, y \in N(C)} [y \in A'] \leq 8d\gamma + 4\gamma \leq 9d\gamma$. By applying Markov's inequality again, for at least $1- \sqrt{\gamma}$ of the cubes $C \in \mathcal{F}$, $g_{x,\sigma}$ and $f_{x,\sigma}$ agree on at least $1-9d\sqrt{\gamma}$ of the points in $C$. Finally, since these cubes are in $\mathcal{F}$, $g_{x,\sigma}$ and $T(C)$ must agree on at least $1 - \sqrt{\gamma} - 9d\sqrt{\gamma} > \frac{d}{q}$ of the points in $C$. By the Schwartz-Zippel lemma, $g_{x,\sigma}|_{C} = T(C)$ for such $C$, and these cubes make up at least $(1-\sqrt{\gamma})^2 > 1-2\sqrt{\gamma}$ of the cubes in $\C_{x,\sigma}$.
\end{proof}

\paragraph{Concluding this section.} For each $(x, \sigma) \in X$ and degree $d$ polynomial $g_{x, \sigma}$, let $\mathcal{F}_{x,\sigma}$ be the set of all cubes $C \in \C_{x,\sigma}$ satisfying $T(C) = g_{x,\sigma|C}$, and $\mathcal{D}_{x,\sigma} = \C_{x,\sigma} \setminus \mathcal{F}_{x,\sigma}$. By Lemma~\ref{lm: gxsigma},  $\mu_x(\mathcal{F}_{x,\sigma}) \geq (1-2\sqrt{\gamma})\mu_x(C_{x,\sigma})$ and $\mu_x(D_{x,\sigma}) \leq \frac{2\sqrt{\gamma}}{1-2\sqrt{\gamma}} \mu_x(\mathcal{F}_{x,\sigma}) \leq \sqrt{5\gamma}\mathcal{F}_{x,\sigma}$.

% For $(x, \sigma) \in X$, 

% let $g_{x,\sigma}$ be the degree $d$ polynomial guaranteed by Lemma~\ref{lm: gxsigma},
% $\mathcal{F}_{x,\sigma}$ be the set of all cubes $C \in \C_{x,\sigma}$ satisfying $T(C) = g_{x,\sigma|C}$, and $\mathcal{D}_{x,\sigma} = \C_{x,\sigma} \setminus \mathcal{F}_{x,\sigma}$. By Lemma~\ref{lm: gxsigma}, $\mu_x(\mathcal{F}_{x,\sigma}) \geq (1-2\sqrt{\gamma})\mu_x(C_{x,\sigma})$ and $\mu_x(D_{x,\sigma}) \leq \frac{2\sqrt{\gamma}}{1-2\sqrt{\gamma}} \mu_x(\mathcal{F}_{x,\sigma}) \leq \sqrt{5\gamma}\mathcal{F}_{x,\sigma}$.

\subsection{Pruning the Pass Probability} \label{sec: prune}
Throughout this section, we consider the distribution $\mathcal{D}$ over tuples $(x, \sigma, C_1, y, \tau, C_2)$ defined by the following process:
\begin{enumerate}
    \item Choose $x \in \mathbb{F}_q^{\num}$ uniformly at random.
    \item Choose $\sigma \in \mathbb{F}_q$ with weight proportional to $\mu_x(\C_{x, \sigma})$.
    \item Choose $C_1 \in \C_{x, \sigma}$ uniformly at random.
    \item Choose $y \in C_1$ uniformly at random and different from $x$
    \item Choose $\tau \in \mathbb{F}_q$ with weight proportional to the fraction of cubes $C \in \C_{x,y}$ with $\dim(C \cap C_1) = 2$ satisfying $T(C)(y) = \tau$.
    \item Choose $C_2$ uniformly at random over $\{C \in \Cxys  \; | \; \dim(C \cap C_1) = 2\}$.
\end{enumerate}

Notice that the marginal distribution over $x,y$ is uniform, and conditioned on $x,y$ 
the marginal distribution over $(C_1,C_2)$ is uniform among all pairs of cubes intersecting in a plane $P$ that contains both $x$ and $y$. 
In particular, the marginal of $(C_1,C_2)$ is uniform among cubes that intersect in a plane, and so letting 
$S = \{(C_1, C_2) \; | \; \dim(C_1 \cap C_2) = 2, \; T(C_1)|_{C_1\cap C_2} = T(C_2)|_{C_1\cap C_2} \}$ we have
\begin{equation}\label{eq:prune1}
   \Pr_{\mathcal{D}}\left[ (C_1, C_2) \in S\right] \geq  \epsilon.
\end{equation}

The distribution $\D$ can also be generated by either first choosing $y \in \Fm$ and proceeding symmetrically, or by first choosing $(x,y) \in \Fm$ and then choosing $(C_1, C_2)$ as a random edge in the induced subgraph of $\AG(3,2)$ on $C_{xy}$ and assigning $(\sigma, \tau)$ accordingly.

The main goal of this section is to prove the following lemma, asserting that we may focus on the cases in which $C_1, C_2$ lie in  
$\mathcal{F}_{\substack{x, \sigma\\ y, \tau}} = \mathcal{F}_{x, \sigma} \cap \mathcal{F}_{y, \tau}$. Namely, 
we show that a constant fraction of the probability on the left hand side of~\eqref{eq:prune1} comes from the event that $C_1,C_2$ are both in $\mathcal{F}_{x, \sigma} \cap \mathcal{F}_{y, \tau}$.
\begin{lemma} \label{lm: pass fxsig}
\begin{equation*}
    \Pr_{\D}\left[(C_1,C_2) \in S \text{ and } C_1, C_2 \in \Fxys\right] \geq \frac{\epsilon}{4}.
\end{equation*}
\end{lemma}

\skipi
The rest of this section is devoted to the proof of Lemma~\ref{lm: pass fxsig}. Towards this end, recall from Section~\ref{sec: local} that for each $(x, \sigma) \in X$ we write $\C_{x,\sigma} = \mathcal{F}_{x,\sigma} \sqcup \mathcal{D}_{x,\sigma}$. Therefore, we can deduce Lemma~\ref{lm: pass fxsig} by showing that the probability of passing the test, 
i.e. $(C_1,C_2)\in S$ in conjunction of either one of the events: (a) $(x, \sigma) \notin X$ and (b) $(x, \sigma) \in X$ but $C_1 \in D_{x,\sigma}$, is small. We begin
by bounding the first event.
\begin{lemma} \label{lm: pass nonexcellent}
 \begin{equation*}
         \Pr_{\D}\left[(C_1, C_2) \in S, \; (x, \sigma) \notin X \right] \leq \frac{\epsilon}{4}.
 \end{equation*}
\end{lemma}
\begin{proof}
Consider the following alternative way of sampling according to $\mathcal{D}$. First, choose a point $x$ uniformly at random, then a random $\sigma \in \Ff_q$ with weight proportional to $\C_{x, \sigma}$, a uniformly at random $C_1 \in \C_{x,\sigma}$, and finally a uniformly at random $C_2$ that intersects $C_1$ in a plane $P \ni x$. The point $y$ can then be chosen at random from $P$ and $\tau = T(C_2)(y)$. Viewing $\D$ this way we see that
\begin{equation*}
    \Pr_{\D}\left[(C_1, C_2) \in S, \; (x, \sigma) \notin X \right]  \leq   
    \E_{x, \sigma}\left[\mathbbm{1}_{(x, \sigma) \notin X} (1-\Phi_{x}(C_{x, \sigma}))\right],
\end{equation*}
since after $x$ and $\sigma$ are fixed, $(C_1, C_2) \in S$ only if $T(C_1)$ and $T(C_2)$ agree on $x$ and $C_2 \in C_{x,\sigma}$. This probability, with $C_{x, \sigma}$ fixed, is precisely $1 - \Phi_x(C_{x,\sigma})$.

If $(x,\sigma) \notin X$, then either $\mu(C_{x, \sigma}) < \epsilon/5$, or $p_{x, \sigma} + \frac{\gamma}{q^{n-2}-2} \geq \gamma \mu(C_{x,\sigma})$ by Lemma~\ref{lm: pxsig}. Therefore,

\begin{align*}
    \E_{x, \sigma}\left[\mathbbm{1}_{(x, \sigma) \notin X} (1-\Phi_{x}(C_{x, \sigma}))\right] \leq
    &\; \E_{x, \sigma}\left[\mathbbm{1}_{\mu_x(C_{x, \sigma})< \epsilon/5} (1-\Phi_{x}(C_{x, \sigma}))\right] \\
    &+ \E_{x, \sigma}\left[\mathbbm{1}_{p_{x, \sigma} > \gamma \left(1-\Phi_{x}(\C_{x,\sigma})\right)}(1-\Phi_{x}(C_{x, \sigma}))\right].
\end{align*}
We bound each term on the right hand side separately. For the first term, applying Lemma~\ref{lm: expansion} yields
\begin{equation*}
    \E_{x, \sigma}\left[\mathbbm{1}_{\mu_x(C_{x, \sigma})< \epsilon/5} (1-\Phi_{x}(C_{x, \sigma}))\right] \leq \frac{\epsilon}{5} + \frac{1}{q}.
\end{equation*}

For the second term, we have
\begin{align*}
    \E_{x, \sigma}\left[\mathbbm{1}_{p_{x, \sigma} > \gamma \left(1-\Phi_{x}(\C_{x,\sigma})\right)}(1-\Phi_{x}(C_{x, \sigma}))\right] 
     &\leq \E_{x, \sigma}\left[\frac{p_{x,\sigma}}{\gamma}\right] \\
     &\leq \frac{1000d^4}{q}.
\end{align*}
The last inequality is due to the fact that $\E_{x, \sigma}\left[p_{x,\sigma} \right]$ can be bounded by the probability that two degree at most $d$ entries $T(C_1)$ and $T(C_2)$ disagree on a line $\ell \subseteq C_1 \cap C_2$, but agree on a point $x \in \ell$. By the Schwartz-Zippel lemma, this probability is at most $d/q$.

Putting everything together gives:
\begin{equation*}
    \Pr_{\D}\left[(C_1, C_2) \in S, \; (x, \sigma) \notin X \right]\leq \frac{\epsilon}{5} + \frac{1000d^4}{q}  \leq \frac{\epsilon}{4}.
\end{equation*}
\end{proof}
As $(x, \sigma)$ and $(y, \tau)$ have the same marginal distribution, we get  
$\Pr_{\D}\left[(C_1, C_2) \in S, (y, \tau) \notin X\right] \leq \epsilon/4$, and it follows by the union bound that
\begin{equation}\label{eq:pass fxsig1}
\Pr_{\D}\left[(C_1, C_2) \in S \land ((x, \sigma) \notin X\lor (y, \tau) \notin X) \right]
\leq \frac{\eps}{2}.
\end{equation}
Next, we bound the contribution from $C_1\in\mathcal{D}_{x,\sigma}$.
\begin{lemma} \label{lm: pass dxsig}
\begin{equation*}
    \Pr_{\D}\left[(C_1, C_2) \in S, \; (x, \sigma) \in X, \; C_1 \in \mathcal{D}_{x, \sigma}\right] \leq 7\gamma \epsilon
\end{equation*}
\end{lemma}
\begin{proof}
First we rewrite the probability as,
\begin{align*}
     \Pr_{\D}\left[(C_1, C_2) \in S, \; (x, \sigma) \in X, \; x \in \mathcal{D}_{x, \sigma}\right]  
     &= \E_{x}\left[ \sum_{(x,\sigma) \in X}  \mu(\mathcal{D}_{x, \sigma}) \Pr_{\D}\left[ (C_1,C_2) \in S \; | \; C_1 \in \mathcal{D}_{x, \sigma}\right]\right].
\end{align*}
We fix $(x,\sigma)\in X$ and upper bound $\Pr_{\D}\left[(C_1,C_2) \in S \; | \;C_1 \in \mathcal{D}_{x, \sigma}\right]$.

Conditioned on $C_1 \in \mathcal{D}_{x,\sigma}$, the distribution $\D$ chooses $C_1 \in \mathcal{D}_{x,\sigma}$ uniformly at random, a plane $P \subseteq C_1$ 
containing $x$, and $C_2$ containing $P$.
For any $C_1 \in \mathcal{D}_{x, \sigma}$, we have by definition that $T(C_1) \neq g_{x, \sigma}|_{C_1}$, so by the Schwartz-Zippel Lemma we have 
$T(C_1)|_{P} = g_{x, \sigma}|_{P}$ with probability at most $d/q$. In the case that $T(C_1)|_{P} \neq g_{x, \sigma}|_{P}$, we immediately get that 
$T(C_1)|_{P} \neq T(C_2)|_{P}$ for any $C_2 \in \mathcal{F}_{x, \sigma}$ that contains $P$ (as $T(C_2) = g_{x,\sigma}|_{C_2}$). 
Furthermore, $T(C_1)|_{P} \neq T(C_2)|_{P}$ for all $C_2 \notin \C_{x, \sigma}$ as $p \ni x$. 
Thus, if the plane $P$ is chosen so that $T(C_1)|_{P} \neq g_{x, \sigma}|_{P}$, then the test can pass only if $C_2 \in \mathcal{D}_{x, \sigma}$. 
Averaged over all $C_1 \in D_{x, \sigma}$ the probability of picking such a plane $p \subseteq C_1$ and then $C_2 \in \mathcal{D}_{x, \sigma}$ is at most $1-\Phi_x(D_{x,\sigma}$). Therefore,
\begin{align*}
    \Pr_{\D}\left[ (C_1,C_2) \in S \; | \;C_1 \in \mathcal{D}_{x, \sigma}\right] \leq  \frac{d}{q}+
    1\cdot(1-\Phi_x(\mathcal{D}_{x,\sigma})) \leq \frac{d+1}{q} + \mu(\mathcal{D}_{x,\sigma}),
\end{align*}
where we use Lemma~\ref{lm: expansion}. Thus,
\begin{align*}
    \E_x \left[\sum_{(x, \sigma) \in X} \mu(\mathcal{D}_{x, \sigma}) \Pr_{\D}\left[(C_1,C_2) \in S \; | \;C_1 \in \mathcal{D}_{x, \sigma}\right] \right] 
    &\leq \E_x\left[ \sum_{(x,\sigma) \in X} \mu(\mathcal{D}_{x,\sigma})^2 \right] + \frac{d+1}{q} \\
    &\leq  \E_x\left[ \sum_{(x,\sigma) \in X} 5\gamma\mu(\mathcal{F}_{x,\sigma})^2 \right] + \frac{d+1}{q}.
\end{align*}
Finally notice that \[
\E_x \left[ 5\gamma\mu(\mathcal{F}_{x,\sigma})^2 \right] \leq 5\gamma \left(\E_x \left[ \mu(\mathcal{F}_{x,\sigma})(1-\Phi_x(\mathcal{F}_{x,\sigma})\right] + \frac{1}{q^{n-2}-2} \right) \leq 5\gamma\left(\epsilon+ \frac{1}{q^{n-2}-2}\right) \leq 6\gamma \epsilon,
\]
where we use the lower bound of Lemma~\ref{lm: expansion} with $k=3$ for the first inequality. The second inequality holds since $\mu(\mathcal{F}_{x,\sigma})(1-\Phi_x(\mathcal{F}_{x,\sigma}))$ is exactly the probability that $C_1$ and $C_2$ are both in $\mathcal{F}_{x,\sigma}$,
in which case the test passes. Hence, when this quantity is summed over all $\sigma$ such that $(x, \sigma) \in X$ and averaged over all $x$, 
the value obtained can be at most the overall pass probability of the test.
\end{proof}

\begin{lemma} \label{lm: D symmetry}
Let $(x, \sigma, C_1, y, \tau, C_2)$ be a tuple chosen from $\D$ with $(C_1, C_2) \in S$. Then, 
\begin{equation*}
    \D(x, \sigma, C_1, y, \tau, C_2) =  \D(y, \tau, C_2, x, \sigma, C_1) = \D(x, \sigma, C_2, y, \tau, C_1) = \D(y, \tau, C_1, x, \sigma, C_2).
\end{equation*}
\end{lemma}
\begin{proof}
Since $(C_1, C_2) \in S$, the table entries $T(C_1)$ and $T(C_2)$ must agree on $x$ and $y$. Thus, $C_1, C_2 \in \mathcal{C}_{x,\sigma}$ and $C_1, C_2 \in \mathcal{C}_{y,\tau}$. Now suppose $\D$ is generated by first choosing cubes $C_1, C_2$ uniformly at random conditioned on $\dim(C_1 \cap C_2) = 2$. Clearly the probability of choosing $(C_1, C_2)$ or $(C_2, C_1)$ is equal. Once the plane $C_1 \cap C_2$ is fixed, the pair $(x,y)$ is chosen uniformly at random over distinct ordered pairs of points in $C_1 \cap C_2$, while $\sigma$ and $\tau$ are set to $T(C_1)(x)$ and $T(C_2)(y)$ respectively. Again, it is clear that the pairs $(x,y)$ and $(y,x)$ have the same probability of being chosen. The result of the lemma then follows.
\end{proof}

By Lemma~\ref{lm: D symmetry}, the bound in Lemma~\ref{lm: pass dxsig} still holds when either (or both) of $(x, \sigma)$ or $C_1$ are replaced by $(y, \tau)$ or $C_2$ respectively. Thus, letting $E$ be the event 
\begin{equation*}
    E = (C_1 \in \mathcal{D}_{x, \sigma}) \lor (C_2 \in \mathcal{D}_{x, \sigma})  \lor (C_1 \in \mathcal{D}_{y, \tau}) \lor (C_2 \in \mathcal{D}_{y, \tau}),
\end{equation*}
it follows that
\begin{equation}\label{eq:pass fxsig2}
\Pr_{\D}\left[(C_1, C_2) \in S \land (x, \sigma) \in X \land (y,\tau)\in X \land E\right] 
\leq 28\gamma \epsilon
\leq \frac{\eps}{4}.
\end{equation}
We are now ready to prove Lemma~\ref{lm: pass fxsig}.
\begin{proof}[Proof of Lemma~\ref{lm: pass fxsig}]
%Combining~\eqref{eq:pass fxsig2},~\eqref{eq:pass fxsig1} and~\eqref{eq:prune1} gives the statement of the lemma.
Let $H$ denote the event that both $(x, \sigma)$ and $(y, \tau)$ are in $X$. By~\eqref{eq:pass fxsig1} and~\eqref{eq:prune1} we get that 
\[
\eps
\leq \Pr_{\D}\left[(C_1, C_2) \in S\right]
= \Pr_{\D}\left[(C_1, C_2) \in S\land H\right] + \Pr_{\D}\left[(C_1, C_2) \in S\land \bar{H}\right]
\leq \Pr_{\D}\left[(C_1, C_2) \in S\land H\right]+\frac{\eps}{2},
\]
so $\Pr_{\D}\left[(C_1, C_2) \in S\land H\right]\geq \frac{\eps}{2}$. 
If $H$ holds, then we can write $\C_{x,\sigma} = \mathcal{F}_{x,\sigma}\sqcup \mathcal{D}_{x,\sigma}$ and $\C_{y, \tau} = \mathcal{F}_{y,\tau} \sqcup \mathcal{D}_{y,\tau}$. Let $H'$ be the event that $C_1, C_2 \in \mathcal{F}_{x,\sigma} \cap \mathcal{F}_{y,\tau}$, then
\[
\frac{\eps}{2}
\leq \Pr_{\D}\left[(C_1, C_2) \in S\land H\right]
= 
\Pr_{\D}\left[(C_1, C_2) \in S\land H\land H'\right]
+
\Pr_{\D}\left[(C_1, C_2) \in S\land H\land E\right],
\]
and by~\eqref{eq:pass fxsig2} we have $\Pr_{\D}\left[(C_1, C_2) \in S\land H\land E\right]\leq \frac{\eps}{4}$ and so 
$\Pr_{\D}\left[(C_1, C_2) \in S\land H\land H'\right]\geq \frac{\eps}{4}$.
\end{proof}
\subsection{Bumps in measure are negligible} \label{sec: bumps}
Throughout this section, let $\mathcal{F}_x$ be an arbitrary set of cubes containing the point $x$ and suppose $\mu_x(\mathcal{F}_x) = \eta$ (in our setting, $\mathcal{F}_x$ will be 
$\mathcal{F}_{x,\sigma}$ for some $(x,\sigma)\in X$). For $c > 1$, let $Y_c = \{ y \in \mathbb{F}_{q}^{\num} \; | \; \mu_{x,y}(\mathcal{F}_x \cap \C_y) \geq c\eta \}$, 
i.e. the set of points $y$ such that $\mathcal{F}_x$ is significantly denser in $\C_{x,y}$. The goal of this section is to show that the mass of $Y_c$ is small, namely:
\begin{lemma}\label{lm: bump up}
Let $Y_c = \{ y \in \Ff_q^n \setminus{x} \; | \; \mu_{x,y}(\mathcal{F}_x \cap \C_y) > c\eta \}$, for $c > 1$. Then, $\mu(Y_c) \leq \frac{4}{(c-1)^2\mu_x(\mathcal{F}_x) q^2}$.
\end{lemma}
Towards the proof of Lemma~\ref{lm: bump up}, consider the distribution $\nu_x$ over $\C_x$ generated by choosing a random $y \in Y_c$ and then a random cube $C$ 
containing $x$ and $y$. Note that $\nu_x(C) = \frac{z(C)}{|Y_c||\C_{x,y}|}$, where $z(C) = |C \cap Y_c|$.

\begin{lemma}\label{lem:var_analysis}
We have:
\begin{enumerate}
  \item $\E_x[z(C)] = \mu(Y_c)(q^3-1)$.
  \item $\var[z(C)] \leq 2\mu(Y_c)q^4$.
\end{enumerate}
\end{lemma}
\begin{proof}
The first item is clear using linearity of expectation, as the probability that any fixed point $y \neq x$ from $Y_c$ is contained in a random cube containing $x$ is $\frac{q^3-1}{|\mathbb{F}_q^n|}$.

For the second item, we write
\begin{equation*}
    \E_{C}[z(C)^2] = \sum_{y_1,y_2 \in Y_c} \E_{C} [\ind_{y_1\in C} \ind_{y_2\in C}].
\end{equation*}
If $x, y_1, y_2$ are contained in a line, then the expectation in the sum is $|Y_c|\frac{q^3-1}{q^n}$. This can happen for at most $|Y_c|q$ pairs $y_1, y_2$. Otherwise the expectation in the sum is $|Y_c|^2\frac{(q^3-1)(q^3-q)}{q^n(q^n-q)}$. Overall this yields
\begin{equation*}
     \E_{C}[z(C)^2] \leq |Y_c|q\frac{q^3-1}{q^n} + |Y_c|^2\frac{(q^3-1)(q^3-q)}{q^n(q^n-q)} \leq \mu(Y_c)q^4 + \mu(Y_c)^2(q^3-1)^2 \frac{q^n}{q^n-q}.
\end{equation*}
Since $\E_x[z(C)] = \mu(Y_c)(q^3-1)$, the second term on the right hand side essentially cancels out and we can crudely bound the variance by
\begin{equation*}
   \var[z(C)] =  \E_x[z(C)^2] - (\mu(Y_c)(q^3-1))^2 \leq \mu(Y_c)q^4 + \mu(Y_c)^2(q^3-1)^2 \frac{q}{q^n-q}
   \leq 2\mu(Y_c)q^4. \qedhere
\end{equation*}
\end{proof}

\begin{proof}[Proof of Lemma~\ref{lm: bump up}]
We have
\begin{equation*}
    |\mu_x(\mathcal{F}_x) - \nu_x(\mathcal{F}_x)|^2 \leq \left(\sum_{C \in \mathcal{F}_x}|\mu_x(C) -\nu_x(C)|\right)^2 \leq |\mathcal{F}_x| \sum_{C \in \mathcal{F}_x} (\mu_x(C) -\nu_x(C))^2,
\end{equation*}
where the first inequality is due to the triangle inequality, and the second inequality is Cauchy-Schwartz. Fix some $y \neq x$, and recall that $\nu_x(C) = \frac{z(C)}{|Y_c||\C_{x,y}|}$
and $\mu_x(C) = \frac{\E_{C}[z(C)]}{|Y_c||\C_{x,y}|} = \frac{1}{|C_x|}$, so
\begin{equation*}
    (\mu_x(C) - \nu_x(C))^2 = \frac{(z(C) - \E[z(C)])^2}{|Y_c|^2|\C_{x,y}|^2}.
\end{equation*}
Therefore we may bound $\sum_{C \in \mathcal{F}_x} (\mu_x(C) -\nu_x(C))^2 \leq\sum_{C \in \C_x} (\mu_x(C) -\nu_x(C))^2 \leq \frac{\var[z(C)]|\C_x|}{|Y_c|^2|\C_{x,y}|^2}$, which with the previous inequality and Lemma~\ref{lem:var_analysis} implies that
\begin{equation*}
     |\mu_x(\mathcal{F}_x) - \nu_x(\mathcal{F}_x)|^2 \leq 2 \frac{|Y_c||\mathcal{F}_x||\C_x|q^4}{|Y_c|^2|\C_{x,y}|^2q^n} \leq \frac{4\eta q^{n} }{|Y_c|q^2} = \frac{4\eta}{\mu(Y_c)q^2}.
\end{equation*}
where we use the facts that $|\C_x|/|\C_{x,y}| \leq \frac{q^{3n} q^3}{q^{2n} q^3(q^3-1)}\leq \frac{q^n}{q^3-1}$ and $|\mathcal{F}_x| = \eta|\C_x|$. By definition, $\nu_x(\mathcal{F}_x) \geq c\eta$, so 
the left hand side is least $(c-1)^2\eta^2$, and combining this with the above inequality yields $\mu(Y_c) \leq \frac{4}{(c-1)^2\eta q^2}$.
\end{proof}
\subsection{Proof of Theorem~\ref{th: main}}
In this section we use the results of Sections~\ref{sec: prune} and~\ref{sec: bumps} to prove Theorem~\ref{th: main}. 
Namely, we find a degree $d$ polynomial that agrees with an $\Omega(\epsilon)$-fraction of the entries in $T$. 

We start by rewriting the probability in Lemma~\ref{lm: pass fxsig} as an expectation over random points $x$ and $y$. This can be seen as generating $\D$ by first choosing $x$ and $y$, and then choosing $(C_1, C_2)$ as two random cubes in $\C_{x,y}$ that intersect in a plane, and setting $\sigma = T(C_1)(x)$ 
and $\tau = T(C_2)(y)$. Lemma~\ref{lm: pass fxsig} then states
\begin{equation*}
   \Pr_{\D}\left[(C_1,C_2) \in S,\; C_1, C_2 \in \mathcal{F}_{x,\sigma} \cap \mathcal{F}_{y,\tau}\right]  = \E_{x,y}\left[ \sum_{(x,\sigma), (y,\tau) \in X} \mu_{x,y}\left( \Fxys \right)\left(1-\Phi_{x,y}\left(\Fxys\right)\right) \right] \geq \frac{\epsilon}{4}.
\end{equation*}
Applying Lemma~\ref{lm: expansion} we get
\begin{equation} \label{eqn}
    \E_{x,y}\left[ \sum_{(x,\sigma), (y,\tau) \in X} \mu_{x,y}\left( \Fxys \right)^2 \right] + \frac{1}{q} \geq \frac{\epsilon}{4}.
\end{equation}

At this point, we would like to restrict to $\Fxys$ such that $\mu_{x,y}\left(\Fxys\right)$ is not significantly larger than $\mu_x(\mathcal{F}_{x,\sigma})$ and $\mu_y(\mathcal{F}_{y,\tau})$. Let $H_{x,\sigma}(y, \tau)$ denote the indicator for the event that $\mu_{x,y}\left(\Fxys\right) > 2\mu_{x}(\mathcal{F}_{x,\sigma})$.
%or $\mu_{x}\left(\Fxys\right) \leq \frac{1}{100}\mu_{x}(\mathcal{F}_{x,\sigma})$.
\begin{lemma}\label{lm: uneven F}
\begin{equation*}
      \E_{x,y}\left[ \sum_{(x,\sigma), (y,\tau) \in X} \mu_{x,y}\left( \Fxys \right)^2 (H_{x,\sigma}(y, \tau) + H_{y, \tau}(x, \sigma)) \right] \leq \frac{400\log(5/\eps)}{q^2}.
\end{equation*}
\end{lemma}
\begin{proof}
We bound
\begin{equation*}
     \E_{x,y}\left[ \sum_{(x,\sigma), (y,\tau) \in X} \mu_{x,y}\left( \Fxys \right)^2 H_{x,\sigma}(y, \tau) \right] = \E_x \left[\sum_{(x,\sigma)\in X}\E_y  \left[ \sum_{(y,\tau)\in X} \mu_{x,y}\left( \Fxys \right)^2 H_{x,\sigma}(y, \tau) \right]\right].
\end{equation*}
Fix an $(x,\sigma) \in X$ and denote $\eta = \mu_{x}(F_{x,\sigma})$. Let $A_y(j)$ denote the indicator of the event that $2^j\eta \leq \mu_{x,y}\left(\Fxy\right) < 2^{j+1} \eta$. Since all cubes in $F_{x,\sigma}$ agree with a function $g_{x,\sigma}$, for each $y$ there can only be one $\tau$, namely $\tau = g_{x,\sigma}(y)$, such that $\mu_{x,y}\left(\Fxys\right) > 0$. Thus, denoting $\Fxy = \Fxys$ for this $\tau$, for all $(x,\sigma) \in X$ we have:
\begin{equation*}
     \E_{y} \left[\sum_{(y,\tau)\in X} \mu_{x, y}\left(\Fxys\right)^2  H_{x,\sigma}(y, \tau) \right] \leq  \E_{y} \left[ \mu_{x, y}(\Fxy)^2  H_{x,\sigma}(y, \tau) \right]
\end{equation*}
We start by performing dyadic partitioning of the inner expectation with $(x,\sigma)$ fixed.
\begin{align*}
    \E_{y} \left[ \mu_{x, y}(\Fxy)^2  H_{x,\sigma}(y, \tau) \right]
     \leq  \sum_{j = 1}^{\log(1/\eta)} \E_y[A_y(j) \mu_{x,y}(\Fxy)^2] 
    \leq  \sum_{j = 1}^{\log(1/\eta)} \E_y[A_y(j) (2^{j+1}\eta)^2]. 
\end{align*}

Using Lemma~\ref{lm: bump up} gives
\begin{equation*}
    \E_y[A_y(j) (2^{j+1}\eta)^2] \leq (2^{j+1}\eta)^2 \mu(Y_{2^j}) \leq \frac{100\eta}{q^2}.
\end{equation*}
Combining everything, we get that
\[
 \E_{x,y}\left[ \sum_{(x,\sigma), (y,\tau) \in X} \mu_{x,y}\left( \Fxys \right)^2 H_{x,\sigma}(y, \tau) \right]
 \leq 
 \E_x \left[\sum_{(x,\sigma)\in X} \log(1/\eta) \frac{ 100\mu(\mathcal{F}_{x,\sigma})}{q^2}\right]
 \leq \frac{200\log(5/\eps)}{q^2}.
\]
Similarly, $\E_{x,y}\left[ \sum_{(x,\sigma), (y,\tau) \in X} \mu_{x,y}\left( \Fxys \right)^2 H_{y,\tau}(x, \sigma) \right]\le  \frac{200\log(5/\eps)}{q^2}$,
and the proof is concluded.
\end{proof}

Let $X'(x,y) = \{(\sigma, \tau) \; | \; (x,\sigma), (y,\tau) \in X, \; H_{x,\sigma}(y,\tau) + H_{y,\tau}(x,\sigma) = 0\}$.

As an immediate consequence of Lemma~\ref{lm: uneven F},
\begin{equation*}
    \E_{x,y}\left[\sum_{(\sigma, \tau) \in X'(x,y)} \mu_{x,y}\left(\Fxys\right)^2 \right] \geq \frac{\epsilon}{4} - \frac{1}{q} -  \frac{400\log(5/\eps)}{q^2} \geq \frac{\epsilon}{5}.
\end{equation*}
Also, 
\[
\E_{x,y}\left[\sum_{(\sigma, \tau) \in X'(x,y)} \mu_{x,y}\left(\Fxys\right)^2 \ind_{\mu_{x,y}\left(\Fxys\right)\leq \frac{\eps}{10}}\right]
\leq
\frac{\eps}{10}\E_{x,y}\left[\sum_{(\sigma, \tau) \in X'(x,y)} \mu_{x,y}\left(\Fxys\right)\right]
\leq \frac{\eps}{10},
\]
and so
\begin{equation*}
    \E_{x,y}\left[\sum_{(\sigma, \tau) \in X'(x,y)} \mu_{x,y}\left(\Fxys\right)^2 \ind_{\mu_{x,y}\left(\Fxys\right)> \frac{\eps}{10}} \right] \geq \frac{\epsilon}{10}.
\end{equation*}

Therefore, there is some $x$ such that
\begin{equation*}
    \E_{y}\left[\sum_{(\sigma, \tau) \in X'(x,y)} \mu_{x,y}\left(\Fxys\right)^2  \ind_{\mu_{x,y}\left(\Fxys\right)> \frac{\eps}{10}}\right]  \geq \frac{\epsilon}{10},
\end{equation*}
and we fix such $x$ henceforth.
For this point $x$, there are at most $5/\epsilon$ values $\sigma$ such that $(\sigma, \tau) \in X'(x,y)$ for some pair $(y, \tau)$. Consequently for this $x$, there is a value $\sigma$ such that
\begin{equation*}
    \E_{y} \left[\sum_{\tau : (\sigma, \tau) \in X'(x,y)}  \mu_{x,y}\left(\Fxys \right)^2 \ind_{\mu_{x,y}\left(\Fxys\right)> \frac{\eps}{10}}\right] \geq \frac{\epsilon^2}{50}.
\end{equation*}
Let $\eta = \mu_x(F_{x,\sigma})$. For each $y$, note that only $\tau = g_{x,\sigma}(y)$ such that $(\sigma, \tau) \in X'(x,y)$ can contribute to the above sum, and the corresponding summand
$\mu_{x,y}\left(\Fxys\right)$ is between $\eps/10$ and $2\eta$, and so there is some $t\in [\eps/10, 2\eta)$ such that
\begin{equation*}
    \E_{y} \left[\mu_{x,y}\left(\Fxys \right)^2 \ind_{\mu_{x,y}\left(\Fxys\right) \in [t,2t]} \ind_{(\sigma, \tau) \in X'(x,y)}\right] \geq \frac{\epsilon^2}{50\log(10/\eps)}.
\end{equation*}
Let $Y$ be the set consisting of $x$ and all $y$'s such that $\mu_{x,y}\left(\Fxys\right) \in [t,2t]$ and $(\sigma, \tau) \in X'(x,y)$ for $\tau = g_{x,\sigma}(y)$; 
writing $p_t = \Pr_{y}\left[y\in Y\right]$, we get that
$t^2 p_t\geq \frac{\epsilon^2}{200\log(10/\eps)}$.

For $y\in Y$, let $\tau = g_{x,\sigma}(y)$ and denote $\mathcal{F}_y = \mathcal{F}_{y, \tau}$ and $\mathcal{F} = \cup_{y\in Y} \mathcal{F}_y$. 

\begin{lemma} \label{lm: equal g}
For all $y\in Y$ we have that $g_{x,\sigma}\equiv g_{y,\tau}$.
\end{lemma}
\begin{proof}
Assume towards contradiction that $g_{x,\sigma}\not\equiv g_{y,\tau}$.
We choose a cube in $C\in \C_{x,y}$ uniformly at random and bound the probability that $g_{x,\sigma}|_{C} = g_{y,\tau}|_{C}$. Note that any cube in $C\in \C_{x,y}$ is of the form $x + \spa(y-x, z_1, z_2)$, so we may choose a cube in $\C_{x,y}$ uniformly at random by choosing $z_1, z_2$, such that $z_1 \notin x+ \spa(y-x)$, $z_2 \notin x + \spa(y-x,z_1)$. By the Schwartz-Zippel Lemma, the probability that $g_{x,\sigma}$ and $g_{y,\tau}$ are equal on both $z_1$ and $z_2$ is at most $(\frac{2d}{q})^2 = \frac{4d^2}{q^2}$. The factors of $2$ account for the fact that $z_1$ and $z_2$ are not uniform over all possible points. Thus,
\[
\frac{\eps}{10}\leq\mu_{x,y}\left(\Fxys\right)\leq \Pr_{C\in \C_{x,y}}\left[g_{x,\sigma}|_{C} = g_{y,\tau}|_{C}\right]\leq \frac{4d^2}{q^2},
\]
and contradiction.
\end{proof}
Thus, denoting $g = g_{x,\sigma}$ we get that $g|_{C} = T(C)$ for all $C\subseteq \mathcal{F}$. We finish the proof by 
showing that $\mu(\mathcal{F})\geq \Omega(\eps)$. Towards this end, consider the bipartite inclusion graph $G_5 = G(\Ff_q^m, \C)$, which has second singular value at most $q^{-3/2}$ by Lemma~\ref{lm: singular value}. By Lemma~\ref{lm: edge sampling},
\begin{equation*}
    \left|\Pr_{y \in Y, C \ni y}\left[C \in \mathcal{F}\right] - \Pr_{C \in \C, y \in C}\left[C \in \mathcal{F}\right]\right| \leq \frac{q^{-3/2}}{\sqrt{\mu(Y)}} = \frac{1}{\sqrt{q^{3}p_t}}
    \leq t\frac{\sqrt{200\log(10/\eps)}}{\sqrt{q^3\eps^2}},
\end{equation*}
where we used $p_t\geq \frac{\epsilon^2}{200t^2\log(10/\eps)}$. Since for each $y \in Y$ we have $(\sigma, \tau) \in X'(x,y)$ for $\tau = g_{x,\sigma}(y)$, we have
\[
\mu_{y}(\mathcal{F}_y) \geq \frac{\mu_{x,y}\left(\Fxys\right)}{2} \geq \frac{t}{2}. 
\] 
Therefore,  $\Pr_{y \in Y, C \ni y}\left[C \in \mathcal{F}\right] \geq \frac{t}{2}$, so we get

\begin{equation*}
    \mu(\mathcal{F}) = \Pr_{C \in \C}\left[C \in \mathcal{F}\right] \geq  \frac{t}{2} - t\frac{\sqrt{200\log(10/\eps)}}{\sqrt{q^3\eps^2}} \geq \frac{t}{3},
\end{equation*}
where we used the fact that $\eps\geq \frac{10^7d^6}{q}$. Since $t\geq \eps/10$, the proof is concluded.

\bibliographystyle{plain}
\bibliography{references}
\end{document}